%% file: main.tex
\documentclass[11pt]{article}

\usepackage{lmodern}

\usepackage{microtype}

\usepackage{amsmath,amssymb,amsthm,mathtools}
\usepackage{booktabs,longtable,tabularx,array}
\usepackage{enumitem}

\usepackage{needspace}

\newcommand{\dom}{\operatorname{dom}}
\newcommand{\DomCrit}{\textnormal{\textsc{DomCrit}}}
\newcommand{\EMinUncol}{\textnormal{\textsc{Edge-Min-3-Uncol}}}
\newcommand{\NP}{\mathrm{NP}}
\newcommand{\coNP}{\mathrm{coNP}}
\newcommand{\DP}{\mathrm{DP}}
\newcommand{\Pclass}{\mathrm{P}}

\newcommand{\ThetaTwoP}{\Theta_2^p}

\newtheorem{theorem}{Theorem}[section]
\newtheorem{lemma}[theorem]{Lemma}
\newtheorem{proposition}[theorem]{Proposition}
\newtheorem{corollary}[theorem]{Corollary}
\theoremstyle{definition}
\newtheorem{definition}[theorem]{Definition}
\newtheorem{remark}[theorem]{Remark}

\usepackage{hyperref}

\usepackage[nameinlink,noabbrev]{cleveref}
\crefname{theorem}{Theorem}{Theorems}
\Crefname{theorem}{Theorem}{Theorems}
\crefname{lemma}{Lemma}{Lemmas}
\Crefname{lemma}{Lemma}{Lemmas}
\crefname{corollary}{Corollary}{Corollaries}
\Crefname{corollary}{Corollary}{Corollaries}
\crefname{section}{Section}{Sections}
\Crefname{section}{Section}{Sections}
\crefname{remark}{Remark}{Remarks}
\Crefname{remark}{Remark}{Remarks}

\hypersetup{hidelinks}

\title{The Complexity of Domatic Criticality}

\author{
Holger Spakowski\\
Department of Mathematics and Applied Mathematics\\
University of Cape Town\\
Rondebosch 7701, South Africa\\
\texttt{Holger.Spakowski@uct.ac.za}
}

\date{}

\begin{document}
\maketitle

\begin{abstract}
The domatic number $\dom(G)$ of a graph $G$ is the maximum number of
dominating sets in a partition of its vertex set.  A graph is
\emph{domatically critical} if deleting any edge lowers its domatic number.
We determine the complexity of recognizing domatically critical graphs both
when the domatic number is prescribed and when it is unrestricted.  The
problems $\DomCrit_1$ and $\DomCrit_2$ are polynomial-time decidable; in
particular, $\DomCrit_2$ consists precisely of the nonempty disjoint unions
of nontrivial stars.  In contrast, for every fixed integer $k\geq3$, the
problem $\DomCrit_k$ is $\DP$-complete under polynomial-time many-one
reductions.  The hardness proof at target value three uses a switch
construction that reduces from edge-minimal $3$-uncolorability and controls
the effect of deleting every edge of the constructed graph.  Clique addition
then lifts the target-three result to every larger fixed target value.
For the unrestricted recognition problem, we prove $\DP$-hardness and
membership in $\ThetaTwoP$.
\end{abstract}

\begin{quote}
\small
\noindent\textbf{Keywords:}
Domatic number; domatic criticality; edge deletion;
graph recognition; $\DP$-completeness; Boolean hierarchy
\end{quote}

\input{introduction}
\input{preliminaries}
\input{edge-criticality}
\input{conclusion}

\bibliographystyle{alpha}
\bibliography{references}

\end{document}

%% file: introduction.tex
\section{Introduction}
\label{sec:introduction}

A dominating set of a graph reaches every vertex, either by containing that
vertex or by containing one of its neighbors.  The \emph{domatic number}
asks how many pairwise disjoint dominating sets can be packed into the same
graph.  Equivalently, it is the largest number of groups into which the
vertices can be partitioned so that every vertex has access to every group in
its closed neighborhood.  Cockayne and Hedetniemi introduced and initiated
the systematic study of this parameter~\cite{CockayneHedetniemi1977}.  Its
standard network interpretation views the parts of such a partition as
separate resource or communication groups, each of which remains locally
accessible throughout the network; see also the later complexity treatment of
Riege and Rothe~\cite{RiegeRothe2006}.

The algorithmic behavior of the domatic number is already well understood in
several directions.  The threshold problem asking whether $\dom(G)\geq k$ is
$\NP$-complete for every fixed $k\geq3$~\cite[Problem~GT3]{GareyJohnson1979}.
Kaplan and Shamir gave a particularly useful reduction from graph coloring to
the domatic-number problem and established hardness on several perfect graph
classes~\cite{KaplanShamir1994}.  Exact-value and comparison questions lead
beyond $\NP$: Riege and Rothe obtained completeness results for $\DP$, higher
levels of the Boolean hierarchy, and parallel access to $\NP$
\cite{RiegeRothe2006,RiegeRothe2006Survey},
and recent work completes the
fixed-value classification for exact domatic number at the remaining values
three and four~\cite{Spakowski2026ExactDomatic}.
These results concern the
value of $\dom(G)$.  The present paper instead asks how robust that value is
under the loss of a single edge.

The established notion of \emph{domatic criticality} is precisely this
edge-deletion sensitivity: a graph $G$ is domatically critical if
$\dom(G-e)<\dom(G)$ for every edge $e\in E(G)$.  Cockayne's 1978 survey
appears already to have proposed the study of this class
\cite[Problem~6]{Cockayne1978Survey}; Zelinka then gave the first dedicated
formal treatment and proved a basic structural theorem for domatically
critical graphs~\cite{Zelinka1980}.  Subsequent work investigated their
relation to domatically full graphs and developed the theory further
\cite{Rall1990,ZverovichZverovich1991}.  In particular, Rall, and later Zverovich and Zverovich, treated the
low-domatic-number structure.  Zverovich and Zverovich also studied the
corresponding cocritical notion and proved that joining a clique preserves and
reflects domatic criticality
\cite[Proposition~4]{ZverovichZverovich1991}.  Thus domatic criticality is a
well-established structural graph-theoretic notion, rather than a new
variation introduced for complexity purposes.

Criticality problems form a natural meeting point of graph structure and the
Boolean hierarchy.  Papadimitriou and Yannakakis introduced $\DP$ and
identified criticality, exact-value, and uniqueness questions as central
sources of problems in this class~\cite{PapadimitriouYannakakis1984}.  The
Boolean hierarchy over $\NP$ was subsequently developed systematically by Cai
et al.~\cite{CaiEtAl1988BooleanI,CaiEtAl1989BooleanII}; for broader background
on exact and minimal graph problems, see also Wagner's work and the survey of
Riege and Rothe~\cite{Wagner1987,Wagner1990,RiegeRothe2006Survey}.  Among the
classical examples, graph minimal uncolorability is $\DP$-complete
\cite{CaiMeyer1987}, as is minimal unsatisfiability
\cite{PapadimitriouWolfe1988}.  More recently, Burjons et al. proved
$\DP$-completeness for edge-minimal $k$-uncolorability and obtained a
$\ThetaTwoP$-completeness result for a vertex-cover criticality problem
\cite{BurjonsEtAl2024}.  The related notion of stability under local graph
modifications has likewise led to $\ThetaTwoP$-complete recognition problems
for several graph parameters~\cite{FreiHemaspaandraRothe2022}.

Despite the structural literature on domatically critical graphs, to the best
of our knowledge their recognition complexity has not previously been
classified.  We determine it for every fixed domatic number and give upper and
lower bounds for the unrestricted problem.  For a fixed $k$, let $\DomCrit_k$
contain the graphs of domatic number exactly $k$ for which every edge deletion
lowers the domatic number to $k-1$, and let $\DomCrit$ denote the corresponding
problem without a prescribed value of $k$.

For every fixed $k\geq2$, the problem $\DomCrit_k$ is precisely the
recognition problem for the class of \emph{domatically $k$-critical graphs}
studied in the structural graph-theory literature
\cite{Rall1990,ZverovichZverovich1991}.  Indeed, deleting a single edge can
lower the domatic number by at most one, so the customary requirement
$\dom(G-e)<\dom(G)=k$ is equivalent to $\dom(G-e)=k-1$ for every
$e\in E(G)$.  We also include the elementary case $k=1$ in order to obtain
a complete fixed-target classification.

The fixed-target classification exhibits a sharp transition.  The case
$k=1$ is immediate: $\DomCrit_1$ consists of the nonempty edgeless graphs.  At
$k=2$, the known structural characterization of Rall, also recorded by
Zverovich and Zverovich, identifies the critical graphs as precisely the
nonempty disjoint unions of nontrivial stars
\cite{Rall1990,ZverovichZverovich1991}.  This yields a linear-time recognition
algorithm.  Starting at the next value, however, the problem becomes complete
for the second level of the Boolean hierarchy:
\[
  \begin{aligned}
    \DomCrit_1,\DomCrit_2 &\in\Pclass,\\
    \DomCrit_k &\text{ is }\DP\text{-complete for every fixed }k\geq3.
  \end{aligned}
\]
For the unrestricted problem we prove
\[
  \DomCrit\text{ is }\DP\text{-hard}
  \qquad\text{and}\qquad
  \DomCrit\in\ThetaTwoP.
\]
The exact complexity of the unrestricted problem remains open; in particular,
we do not obtain $\ThetaTwoP$-hardness.

The main technical contribution is the hardness proof at target value three.
We reduce from edge-minimal $3$-uncolorability, shown $\DP$-complete by
Burjons et al.~\cite[Theorem~8]{BurjonsEtAl2024}.  The local
encoding begins with the degree-two forcing principle that also underlies the
Kaplan--Shamir reduction~\cite{KaplanShamir1994}: in a $3$-domatic coloring,
a vertex of degree two and its two neighbors must receive the three distinct
colors.  Hence a subdivided edge can enforce the inequality constraint of a
proper $3$-coloring.

Criticality requires additional control.  For each edge $e$ of the input graph
$G$, we build a switch component $C(G,e)$.  Private triangles ensure that every copy of an
original vertex sees all three colors, while a distinguished switch edge
selectively disables the coloring constraint represented by $e$.  The
component is $3$-domatic exactly when $G-e$ is $3$-colorable; after deletion
of its switch edge, it is $3$-domatic exactly when $G$ itself is
$3$-colorable.  Taking the disjoint union of one component $C(G,e)$ for each edge
$e\in E(G)$ converts these local correspondences into the universal condition
in domatic criticality.  The same construction also verifies that deletion of
every non-switch edge lowers the domatic number, and therefore yields a single
reduction that proves both $\DP$-completeness of $\DomCrit_3$ and $\DP$-hardness
of unrestricted $\DomCrit$.  The private-triangle and switch mechanisms are
specific to the present reduction; only the underlying degree-two forcing idea
is shared with the Kaplan--Shamir construction.

Finally, the target-three result extends uniformly to all fixed values
$k\geq3$.  The classical identity
$\dom(G+K_r)=\dom(G)+r$, originating with Cockayne and Hedetniemi
\cite{CockayneHedetniemi1977}, combines with the clique-join theorem of
Zverovich and Zverovich~\cite[Proposition~4]{ZverovichZverovich1991} to lift
criticality from $G$ to $G+K_r$.  This produces the complete fixed-target
classification.  For unrestricted recognition, standard domatic-threshold
queries give the $\ThetaTwoP$ upper bound.

The remainder of the paper is organized as follows.
Section~\ref{sec:preliminaries} introduces domatic colorings, the recognition
problems, the relevant complexity classes, and the elementary structural facts
used later.  Section~\ref{sec:edge-criticality} proves the target-two
characterization, the fixed-target upper bound, the switch reduction at target
three, the lifting theorem, and the bounds for the unrestricted problem.
Section~\ref{sec:conclusion} concludes with the remaining open questions.

%% file: preliminaries.tex
\section{Preliminaries}
\label{sec:preliminaries}

\subsection{Graphs, domination, and domatic colorings}
\label{subsec:domatic-definitions}

All graphs considered in this paper are finite, simple, and undirected.  For a
 graph $G$, we write $V(G)$ and $E(G)$ for its vertex and edge sets.  For
$v\in V(G)$, the open and closed neighborhoods of $v$ are
\[
  N_G(v)=\{u\in V(G):\{u,v\}\in E(G)\}
  \qquad\text{and}\qquad
  N_G[v]=N_G(v)\cup\{v\},
\]
respectively.  We write $\deg_G(v)=|N_G(v)|$ and, for nonempty $G$,
\[
  \delta(G)=\min_{v\in V(G)}\deg_G(v).
\]
A vertex of degree zero is \emph{isolated}.  The graph obtained by deleting
an edge $e\in E(G)$ is denoted by $G-e$.

The disjoint union of graphs $G_1,\ldots,G_t$ is denoted by
\[
  \mathop{\dot\bigcup}_{i=1}^{t}G_i.
\]
When the given graphs are not already vertex-disjoint, this notation refers to
vertex-disjoint isomorphic copies.  For vertex-disjoint graphs $G$ and $H$,
their \emph{join}, denoted by $G+H$, is obtained from their disjoint union by
adding every edge with one endpoint in $V(G)$ and the other in $V(H)$.  We
write $K_0$ for the empty graph, so that $G+K_0=G$.

A set $D\subseteq V(G)$ is a \emph{dominating set} of $G$ if
$D\cap N_G[v]\neq\varnothing$ for every $v\in V(G)$.  A \emph{$k$-domatic
partition} of $G$ is a partition of $V(G)$ into $k$ dominating sets.  The
domatic number, introduced and systematically studied by Cockayne and
Hedetniemi~\cite{CockayneHedetniemi1977}, is the largest number of classes in
such a partition.

It is often useful to regard a domatic partition as a coloring.  This
viewpoint makes the local constraints in our constructions particularly
transparent.

\begin{definition}[Domatic coloring]
\label{def:domatic-coloring}
Let $G$ be nonempty and let $k\geq1$.  A \emph{$k$-domatic coloring} of
$G$ is a function
\[
  c:V(G)\longrightarrow[k],
  \qquad [k]=\{1,\ldots,k\},
\]
such that
\[
  c(N_G[v])=[k]
  \qquad\text{for every }v\in V(G).
\]
Thus every vertex sees every color in its closed neighborhood.  Equivalently,
each color class is a dominating set.  The \emph{domatic number} $\dom(G)$
is the largest $k$ for which $G$ has a $k$-domatic coloring.
\end{definition}

For the degenerate case arising when the sole vertex of $K_1$ is deleted, we
set
\[
  \dom(\varnothing)=0.
\]
If several color classes of a domatic coloring are identified, the resulting
coloring remains domatic.  More precisely, if $c:V(G)\to[k]$ is domatic and
$\pi:[k]\to[\ell]$ is surjective, then $\pi\circ c$ is an $\ell$-domatic
coloring.

A \emph{proper $k$-coloring} of a graph $G$ is a function
$\gamma:V(G)\to[k]$ such that adjacent vertices receive different colors.  A
graph is \emph{$k$-colorable} if it has a proper $k$-coloring.

\subsection{Domatic criticality}
\label{subsec:criticality-definitions}

The established notion of domatic criticality is defined by edge deletion.
Following Zelinka, Rall, and Zverovich and
Zverovich~\cite{Zelinka1980,Rall1990,ZverovichZverovich1991}, a graph is
\emph{domatically critical} if deleting any one of its edges lowers its domatic
number.  A domatically critical graph of domatic number $k$ is also called
\emph{domatically $k$-critical}.  This paper studies the computational
complexity of recognizing these graphs.

Proposition~\ref{prop:edge-deletion-bound} shows that, for every edge
$e\in E(G)$,
\[
  \dom(G)-1\leq\dom(G-e)\leq\dom(G).
\]
Since the domatic number is integer-valued, the following three conditions are
equivalent:
\[
\begin{aligned}
  \dom(G-e)<\dom(G)
    &\Longleftrightarrow \dom(G-e)\leq\dom(G)-1,\\
    &\Longleftrightarrow \dom(G-e)=\dom(G)-1.
\end{aligned}
\]
We shall use these strict-decrease, upper-bound, and exact-decrease
formulations interchangeably.

\begin{definition}[Edge-deletion domatic criticality]
\label{def:edge-criticality}
For every fixed integer $k\geq1$, define
\[
  \DomCrit_k
  =
  \left\{
    G\ \middle|\
    \begin{array}{l}
      \dom(G)=k,\\
      \dom(G-e)=k-1\text{ for every }e\in E(G)
    \end{array}
  \right\}.
\]
Equivalently,
\[
  \DomCrit_k
  =
  \left\{
    G\ \middle|\
    \begin{array}{l}
      \dom(G)=k,\\
      \dom(G-e)\leq k-1\text{ for every }e\in E(G)
    \end{array}
  \right\}.
\]
The unrestricted recognition problem is
\[
  \DomCrit
  =
  \left\{
    G\ \middle|\
    \dom(G-e)=\dom(G)-1
    \text{ for every }e\in E(G)
  \right\}.
\]
Equivalently,
\[
  \DomCrit
  =
  \left\{
    G\ \middle|\
    \dom(G-e)\leq\dom(G)-1
    \text{ for every }e\in E(G)
  \right\}.
\]
\end{definition}

We use the usual convention that a universally quantified statement over an
empty set is true.  Consequently, an edgeless graph satisfies the deletion
condition in Definition~\ref{def:edge-criticality}; membership in a fixed-target
class $\DomCrit_k$ still additionally requires $\dom(G)=k$.

\subsection{Complexity-theoretic background}
\label{subsec:complexity-background}

All reductions in this paper are polynomial-time many-one reductions.  We write
$A\leq_m^p B$ if there is a polynomial-time computable function $f$ such that
$x\in A$ if and only if $f(x)\in B$.

For complexity classes $\mathcal C$ and $\mathcal D$, let
\[
  \mathcal C\wedge\mathcal D
  =
  \{A\cap B:A\in\mathcal C\text{ and }B\in\mathcal D\}.
\]
Papadimitriou and Yannakakis introduced the class
\[
  \DP=\NP\wedge\coNP
      =\{A-B:A,B\in\NP\}
\]
and identified exact-value, uniqueness, and criticality questions as natural
sources of problems in this class~\cite{PapadimitriouYannakakis1984}.  The
class $\DP$ is the second level of the Boolean hierarchy over
$\NP$~\cite{CaiEtAl1988BooleanI,CaiEtAl1989BooleanII}.

The unrestricted recognition problem will be shown to lie in
\(\Theta_2^p\).  This class admits the standard equivalent
characterizations
\[
  \Theta_2^p
  =
  \mathrm{P}^{\mathrm{NP}}_{\parallel}
  =
  \mathrm{P}^{\mathrm{NP}[O(\log n)]},
\]
where the first oracle model permits polynomially many nonadaptive
\(\mathrm{NP}\) queries and the second permits \(O(\log n)\) adaptive
\(\mathrm{NP}\) queries
\cite{Hemachandra1987,BussHay1988,KoblerSchoningWagner1987,Wagner1990}.

In the domatic-number setting, exact and comparison variants already
yield natural complete problems in \(\mathrm{DP}\), higher levels of the
Boolean hierarchy, and \(\Theta_2^p\)~\cite{RiegeRothe2006}.

The basic decision version of the domatic-number problem asks whether a given
graph admits a domatic partition of a prescribed size.  For a fixed integer
$k\geq1$, write
\[
  k\text{-}\mathrm{DNP}
  =\{G:\dom(G)\geq k\}.
\]
Membership in $\NP$ is witnessed by a $k$-domatic coloring.  It is classical
that $k\text{-}\mathrm{DNP}$ is $\NP$-complete for every fixed
$k\geq3$; see Garey and Johnson~\cite[Problem~GT3]{GareyJohnson1979}, who
attribute the underlying reduction to unpublished work of Garey, Johnson, and
Tarjan.

\subsection{Basic properties of the domatic number}
\label{subsec:elementary-facts}

We shall use the following classical and elementary facts throughout the paper.
Their short proofs are given so that the later arguments can refer to them
directly.

\Needspace{9\baselineskip}
The classical minimum-degree bound is due to Cockayne and
Hedetniemi~\cite{CockayneHedetniemi1977}.

\begin{proposition}[Minimum-degree bound]
\label{prop:degree-bound}
For every nonempty graph $G$,
\[
  \dom(G)\leq\delta(G)+1.
\]
\end{proposition}

\begin{proof}
Let $v$ be a vertex of minimum degree.  Every color in a domatic coloring
must occur in $N_G[v]$.  Since $|N_G[v]|=\delta(G)+1$, there can be at most
$\delta(G)+1$ colors.
\end{proof}

A classical theorem attributed to Ore states that, in a graph without isolated
vertices, the complement of every inclusion-minimal dominating set is also
dominating; see Cockayne~\cite[Proposition~1]{Cockayne1978Survey}.  It yields
the following characterization.

\begin{proposition}[Isolated vertices and two domatic colors]
\label{prop:one-two-domatic}
A nonempty graph $G$ has domatic number one if and only if it has an isolated
vertex.  Equivalently,
\[
  \dom(G)\geq2
  \quad\Longleftrightarrow\quad
  G\text{ has no isolated vertex}.
\]
\end{proposition}

\begin{proof}
Suppose first that $v$ is isolated.  Then $N_G[v]=\{v\}$, so no domatic
coloring can make $v$ see two colors.  Since every nonempty graph has a
one-color domatic coloring, $\dom(G)=1$.

Conversely, suppose that $G$ has no isolated vertex.  Choose an
inclusion-minimal dominating set $D$.  By the theorem attributed to Ore,
$V(G)\setminus D$ is also dominating.  Hence
\[
  V(G)=D\mathbin{\dot\cup}\bigl(V(G)\setminus D\bigr)
\]
is a partition into two dominating sets, and therefore $\dom(G)\geq2$.
\end{proof}

The following local observation is central to our hardness reduction.

\begin{lemma}[Degree-two forcing]
\label{lem:degree-two-forcing}
Let $c$ be a $3$-domatic coloring of $G$, and let $x$ be a vertex of degree
two with neighbors $u$ and $v$.  Then $c(u)$, $c(x)$, and $c(v)$ are the
three distinct colors.  In particular, $c(u)\neq c(v)$.
\end{lemma}

\begin{proof}
The closed neighborhood $N_G[x]=\{u,x,v\}$ has three vertices and must
contain all three colors.
\end{proof}

The domatic number behaves componentwise under disjoint union.

\begin{lemma}[Disjoint-union formula]
\label{lem:disjoint-union}
Let $G_1,\ldots,G_t$ be nonempty graphs, where $t\geq1$.  Then
\[
  \dom\!\left(\mathop{\dot\bigcup}_{i=1}^{t}G_i\right)
  =
  \min_{1\leq i\leq t}\dom(G_i).
\]
\end{lemma}

\begin{proof}
Put $G=\dot\bigcup_{i=1}^{t}G_i$.  If $c$ is an $\ell$-domatic coloring of
$G$, then $N_G[v]=N_{G_i}[v]$ for every $v\in V(G_i)$.  The restriction of
$c$ to each $G_i$ is therefore $\ell$-domatic, and hence
\[
  \dom(G)\leq\min_i\dom(G_i).
\]

Conversely, let $k=\min_i\dom(G_i)$.  For each $i$, identify color classes
of a maximum domatic coloring of $G_i$, if necessary, to obtain a
$k$-domatic coloring of $G_i$.  Combining these colorings gives a
$k$-domatic coloring of $G$.
\end{proof}

The following clique-addition identity is another classical result of Cockayne
and Hedetniemi~\cite{CockayneHedetniemi1977}.

\begin{proposition}[Clique-addition identity]
\label{prop:clique-join}
For every graph $G$ and every integer $r\geq0$,
\[
  \dom(G+K_r)=\dom(G)+r.
\]
\end{proposition}

\begin{proof}
If $G=\varnothing$, then $G+K_r=K_r$ and
$\dom(K_r)=r=\dom(G)+r$.  Assume that $G$ is nonempty and put
$d=\dom(G)$.  Take a $d$-domatic coloring of $G$ and give the $r$ vertices
of $K_r$ pairwise distinct new colors.  Every vertex of $G$ sees the old
colors within $G$ and all new colors on the clique, while every clique
vertex is universal.  Hence $\dom(G+K_r)\geq d+r$.

For the reverse inequality, let $c:V(G+K_r)\to[t]$ be a $t$-domatic
coloring.  Put $Q=V(K_r)$, $S=c(Q)$, $s=|S|$, and
$T=[t]\setminus S$.  Then $s\leq r$.  If $T=\varnothing$, then
$t=s\leq r\leq d+r$.  Suppose that $T\neq\varnothing$.  For every
$v\in V(G)$ and every color $\alpha\in T$, the color $\alpha$ must occur
in $N_G[v]$, because it does not occur on $Q$.  Recolor every vertex of
$G$ whose color lies in $S$ with one fixed color from $T$, leaving all
colors in $T$ unchanged.  After renaming the colors, this is a
$|T|$-domatic coloring of $G$.  Thus $t-s=|T|\leq d$, and therefore
$t\leq d+r$.
\end{proof}

Rall proved the following bound for the deletion of a single
edge~\cite[Proposition~3]{Rall1990}.

\begin{proposition}[Edge-deletion bound]
\label{prop:edge-deletion-bound}
For every graph $G$ and every edge $e\in E(G)$,
\[
  \dom(G)-1\leq\dom(G-e)\leq\dom(G).
\]
\end{proposition}

\begin{proof}
Every domatic coloring of $G-e$ is also a domatic coloring of $G$, since
adding an edge can only enlarge closed neighborhoods.  Hence
$\dom(G-e)\leq\dom(G)$.

For the reverse inequality, put $d=\dom(G)$, choose a $d$-domatic coloring
$c$, and write $e=xy$.  If $c(x)=c(y)$, then $c$ remains domatic after
$e$ is deleted: only the closed neighborhoods of $x$ and $y$ change, and
each endpoint still sees their common color on itself.

Suppose that $c(x)\neq c(y)$.  Identify the colors $c(x)$ and $c(y)$ and
leave all other colors distinct.  Every vertex other than $x$ and $y$ has
the same closed neighborhood after the deletion.  The vertex $x$ may lose
the old color of $y$, but that color has been identified with $c(x)$,
which still occurs at $x$; the argument for $y$ is symmetric.  The resulting
coloring is $(d-1)$-domatic on $G-e$.
\end{proof}

\begin{corollary}
\label{cor:edge-exact}
Deleting an edge lowers the domatic number if and only if it lowers it by
exactly one.
\end{corollary}

\subsection{Edge-minimal three-uncolorability}
\label{subsec:source-problems}

Our hardness reduction starts from edge-minimal three-uncolorability.

\begin{definition}[Edge-minimal three-uncolorability]
\label{def:minimal-uncolorability}
Define
\[
  \EMinUncol
  =
  \left\{
    G\ \middle|\
    \begin{array}{l}
      G\text{ is not $3$-colorable, and}\\
      G-e\text{ is $3$-colorable}\\
      \text{for every }e\in E(G)
    \end{array}
  \right\}.
\]
No connectedness assumption is imposed.
\end{definition}

Burjons et~al. proved that edge-minimal $k$-uncolorability is
$\DP$-complete for every fixed $k\geq3$~\cite[Theorem~8]{BurjonsEtAl2024}.
We shall use the case of three colors.

\begin{theorem}[Burjons et al.]
\label{thm:source-problems}
The problem $\EMinUncol$ is $\DP$-complete under polynomial-time many-one
reductions.
\end{theorem}

%% file: edge-criticality.tex
\section{Domatic criticality}
\label{sec:edge-criticality}

In this section, we determine the complexity of recognizing domatically
critical graphs.  We first settle the case $k=2$, then prove
$\DP$-completeness for every fixed $k\geq3$, and finally
establish upper and lower bounds for the unrestricted problem.

\Needspace{8\baselineskip}
\subsection{The case \texorpdfstring{{\boldmath $k=2$}}{k=2}}
\label{subsec:edge-target-two}

Rall characterized the domatically critical graphs of domatic number two as
precisely the nonempty disjoint unions of nontrivial stars~\cite{Rall1990}.
Zverovich and Zverovich later also recorded this
characterization~\cite[p.~278]{ZverovichZverovich1991}.  We include a direct
proof because the argument is short and also yields an immediate recognition
algorithm.  A
\emph{nontrivial star} is a graph $K_{1,r}$ with $r\geq1$.

\begin{lemma}
\label{lem:edge-critical-edge-leaf}
Let $H\in\DomCrit_2$.  Every edge of $H$ has an endpoint of degree one.
\end{lemma}

\begin{proof}
Fix an edge $e=xy$.  Since $H\in\DomCrit_2$, we have
\[
  \dom(H-e)=1.
\]
By Proposition~\ref{prop:one-two-domatic}, the graph $H-e$ has an isolated
vertex.  Only the degrees of $x$ and $y$ change when $e$ is deleted, so one of
these two vertices is isolated in $H-e$.  The corresponding endpoint has
degree one in $H$.
\end{proof}

\begin{lemma}
\label{lem:edge-star-components}
Let $H$ be a connected graph with at least one edge.  If every edge of $H$
has an endpoint of degree one, then $H$ is a nontrivial star.
\end{lemma}

\begin{proof}
Choose an edge $xy$.  If both $x$ and $y$ have degree one, connectedness gives
$H=K_2=K_{1,1}$.  Otherwise, after interchanging the endpoints if necessary,
assume that $\deg_H(x)\geq2$.  The hypothesis then gives $\deg_H(y)=1$.

We claim that every vertex different from $x$ is adjacent to $x$.  Suppose
otherwise, and let
\[
  x=v_0,v_1,\ldots,v_t=w
\]
be a shortest path from $x$ to a non-neighbor $w$ of $x$.  Then $t\geq2$.
Because $\deg_H(x)\geq2$, the edge $xv_1$ can have an endpoint of degree one
only if $\deg_H(v_1)=1$.  This contradicts the further adjacency of $v_1$ to
$v_2$.  Thus every vertex other than $x$ is adjacent to $x$.  Applying the
hypothesis to each edge incident with $x$, and again using
$\deg_H(x)\geq2$, shows that all other vertices have degree one.  Hence
$H=K_{1,r}$ for some $r\geq2$.
\end{proof}

\begin{theorem}[{\cite{Rall1990}}]
\label{thm:edge-target-two-characterization}
A graph $H$ is in $\DomCrit_2$ if and only if it is a nonempty disjoint
union of nontrivial stars.
\end{theorem}

\begin{proof}
Suppose first that $H\in\DomCrit_2$.  Since $\dom(H)=2$,
Proposition~\ref{prop:one-two-domatic} implies that $H$ has no isolated
vertex.  Thus every connected component contains an edge.  By
Lemmas~\ref{lem:edge-critical-edge-leaf} and~\ref{lem:edge-star-components}, every component is a nontrivial star.

Conversely, suppose that $H$ is a nonempty disjoint union of nontrivial
stars.  It has no isolated vertex, so
$\dom(H)\geq2$ by Proposition~\ref{prop:one-two-domatic}.  Every component
has a leaf, and hence $\delta(H)=1$.  The minimum-degree bound gives
$\dom(H)\leq2$, so $\dom(H)=2$.

Let $e$ be any edge of $H$.  Deleting $e$ isolates at least one endpoint in
the star component containing $e$; if that component is $K_2$, both endpoints
become isolated.  Therefore $H-e$ has an isolated vertex, and
Proposition~\ref{prop:one-two-domatic} gives $\dom(H-e)=1$.  Thus
$H\in\DomCrit_2$.
\end{proof}

\begin{corollary}
\label{cor:edge-target-two-p}
The problem $\DomCrit_2$ is in $\Pclass$.
\end{corollary}

\begin{proof}
Reject an empty input graph or a graph with an isolated vertex.  For every
edge $xy$, test whether $\deg(x)=1$ or $\deg(y)=1$.  The tests all succeed
exactly for the nonempty disjoint unions of nontrivial stars, by
Theorem~\ref{thm:edge-target-two-characterization}.  With adjacency lists,
the procedure runs in linear time.
\end{proof}

\subsection{The fixed-target upper bound}
\label{subsec:edge-fixed-membership}

For fixed $k$, membership separates into two threshold conditions: the input
graph must be $k$-domatic, and no edge-deleted graph may remain $k$-domatic.
The edge-deletion bound then identifies their intersection with
$\DomCrit_k$.

\begin{proposition}
\label{prop:edge-fixed-membership}
For every fixed integer $k\geq1$, the problem $\DomCrit_k$ is in $\DP$.
\end{proposition}

\begin{proof}
For $k=1$, the problem consists exactly of the nonempty edgeless graphs.  Indeed,
a nonempty edgeless graph has domatic number one and satisfies the universal
edge-deletion condition vacuously.  Conversely, if $H$ has an edge, then every
edge-deleted graph is nonempty and therefore has domatic number at least one,
so it cannot have domatic number zero.  Thus $\DomCrit_1\in\Pclass\subseteq
\DP$.

Fix $k\geq2$, and let
\[
  A_k
  =
  \{H\mid \dom(H)\geq k\}
  =k\text{-}\mathrm{DNP}.
\]
By the definition of the domatic-number threshold problem, $A_k\in\NP$;
a certificate is a $k$-domatic coloring of $H$.

Let
\[
  B_k
  =
  \left\{
    H\ \middle|\
    \dom(H-e)\leq k-1
    \text{ for every }e\in E(H)
  \right\}.
\]
The complement of $B_k$ is in $\NP$: one guesses an edge $e\in E(H)$
and a $k$-domatic coloring of $H-e$.  Hence $B_k\in\coNP$.

We have
\[
  \DomCrit_k=A_k\cap B_k.
\]
The forward inclusion follows immediately from the equivalent upper-bound
formulation in Definition~\ref{def:edge-criticality}.  Conversely, let
$H\in A_k\cap B_k$.  Since $\dom(H)\geq k\geq2$, the graph $H$ has an
edge.  For any $e\in E(H)$, Proposition~\ref{prop:edge-deletion-bound} gives
\[
  k\leq\dom(H)\leq\dom(H-e)+1\leq k.
\]
Hence $\dom(H)=k$, and the condition $H\in B_k$ is precisely the upper-bound
criticality condition from Definition~\ref{def:edge-criticality}.  Thus
$H\in\DomCrit_k$.  Therefore
\[
  \DomCrit_k=A_k\cap B_k\in\NP\wedge\coNP=\DP.
\qedhere
\]
\end{proof}

\subsection{\texorpdfstring{{\boldmath $\DP$}}{DP}-completeness for \texorpdfstring{{\boldmath $k=3$}}{k=3}}
\label{subsec:edge-target-three}

\subsubsection{Idea of the reduction from edge-minimal \texorpdfstring{{\boldmath $3$}}{3}-uncolorability}
\label{subsec:edge-construction-idea}

To prove DP-hardness for $k=3$, we use the edge-minimal
$3$-uncolorability problem $\EMinUncol$ from
Definition~\ref{def:minimal-uncolorability}.  An input graph $G$ is a
yes-instance precisely when $G$ is not $3$-colorable but $G-e$ is
$3$-colorable for every edge $e\in E(G)$.  This problem is
DP-complete~\cite[Theorem~8]{BurjonsEtAl2024}.

Every edgeless graph is $3$-colorable and hence is a no-instance.  Replacing an edgeless input by the fixed no-instance $K_2$
therefore shows that $\EMinUncol$ remains DP-complete when restricted to
graphs with at least one edge.  We use this restriction throughout the
present reduction.

We construct a polynomial-time many-one reduction $R$ from
this restriction of $\EMinUncol$
and prove the stronger equivalence
\[
\begin{aligned}
  G\in\EMinUncol
    &\Longleftrightarrow R(G)\in\DomCrit_3,\\
  R(G)\in\DomCrit_3
    &\Longleftrightarrow R(G)\in\DomCrit.
\end{aligned}
\]
Thus every yes-instance is mapped to a domatically critical graph of domatic
number three, whereas every no-instance is mapped to a graph that is not
domatically critical.  The quantifier pattern of $\EMinUncol$ matches domatic
edge criticality particularly closely: the universal condition that every
graph $G-e$ be $3$-colorable is represented by the one-component-per-edge
disjoint union, while the non-$3$-colorability of $G$ is detected when the
switch edge in any component is deleted.

The global form of the reduction is simple.  For every edge $e\in E(G)$, we
construct a graph $C(G,e)$ with a distinguished switch edge $\sigma_e$, and
we take the disjoint union of one such component for every edge of $G$:
\[
  R(G)=\mathop{\dot\bigcup}_{e\in E(G)} C(G,e).
\]
Thus the reduction associates one entire switch component with each edge $e$
of the input graph.  The component $C(G,e)$ is designed to encode the
$3$-colorability of $G-e$, whereas deleting its switch edge makes it encode
the $3$-colorability of $G$.

The construction implements this correspondence by adapting the local encoding
principle underlying the reduction of Kaplan and
Shamir~\cite{KaplanShamir1994}, which translates proper colorings into domatic
colorings.  In a $3$-domatic coloring, a vertex of degree two has a closed
neighborhood of exactly three vertices; those three vertices must consequently
receive the three different colors.  In particular, the two neighbors of the
degree-two vertex have different colors.  Subdividing an edge therefore turns
the corresponding proper-coloring constraint into a local domatic-coloring
constraint.

The construction first enlarges $G$ by placing every vertex $v\in V(G)$ in a
private triangle.  We refer to all vertices of this enlarged graph---the
vertices inherited from $G$ together with the two new vertices in each private
triangle---as \emph{core vertices}.  Fresh copies of these core vertices form
the main part of each component $C(G,e)$; the edges between them are then
replaced by subdivision paths, except that the distinguished edge $e$ is
represented by a switch path.

Criticality requires substantially more than this local encoding.  First, every
core vertex must itself see all three colors.  The private triangles ensure
this: after their edges are subdivided, the two subdivision vertices on the
private-triangle edges incident with a given core vertex supply the two colors
different from the color of that core vertex.  Second, for the distinguished
edge $e$, the switch gadget must allow the endpoints of $e$ to receive the same
color while $\sigma_e$ is present, but force them to receive different colors
after $\sigma_e$ is deleted.  Thus $C(G,e)$ represents all proper-coloring
constraints except possibly the one corresponding to $e$, whereas
$C(G,e)-\sigma_e$ represents all of them.

The two states of a switch component are summarized as follows:
\[
\begin{array}{c|c|c}
  \text{component} & \text{constraint represented by }e
    & \text{$3$-domatic exactly when}\\
  \hline
  C(G,e) & \text{disabled} & G-e\text{ is $3$-colorable}\\
  C(G,e)-\sigma_e & \text{restored} & G\text{ is $3$-colorable}.
\end{array}
\]
Finally, the disjoint union over all edges converts this local switch behavior
into edge criticality.  Since the domatic number of a disjoint union is the
minimum of the component domatic numbers, $R(G)$ is $3$-domatic exactly when
every component $C(G,e)$ is $3$-domatic.  Moreover, an edge deletion that lowers
the domatic number of any one component lowers the domatic number of the entire
union.  This is why the reduction uses one component for every edge of the
input graph.  The private triangles and the switch gadget are the new
ingredients needed for this criticality reduction; only the underlying
degree-two forcing principle is shared with the Kaplan--Shamir construction.

\subsubsection{Formal construction}
\label{subsec:edge-formal-construction}

Let $G=(V,E)$ be an instance of $\EMinUncol$.  We first saturate its vertices
by private triangles.

\begin{definition}[Triangle saturation]
\label{def:edge-triangle-saturation}
For every $v\in V$, introduce two fresh vertices $v^1$ and $v^2$ and add the
three edges
\[
  vv^1,\qquad v^1v^2,\qquad v^2v.
\]
All original edges of $G$ are retained.  The resulting graph is denoted by
$G^\triangle$; thus
\[
  V(G^\triangle)
  =V\cup\{v^1,v^2:v\in V\}
\]
and
\[
  E(G^\triangle)
  =E\cup
  \bigcup_{v\in V}
  \{vv^1,v^1v^2,v^2v\}.
\]
The triangle induced by $\{v,v^1,v^2\}$ is the \emph{private triangle} of
$v$.  We call all vertices of $G^\triangle$ \emph{core vertices}; the vertices
in $V$ are the \emph{original vertices}, and the vertices $v^1,v^2$ are the
\emph{preprocessing vertices}.
\end{definition}

\begin{lemma}
\label{lem:edge-saturation-colorability}
For every original edge $e\in E(G)$,
\[
  G\text{ is $3$-colorable}
  \quad\Longleftrightarrow\quad
  G^\triangle\text{ is $3$-colorable},
\]
and
\[
  G-e\text{ is $3$-colorable}
  \quad\Longleftrightarrow\quad
  G^\triangle-e\text{ is $3$-colorable}.
\]
\end{lemma}

\begin{proof}
Every proper $3$-coloring of $G$ extends independently over each private
triangle: after coloring $v$, assign the two remaining colors to $v^1$ and
$v^2$.  The same extension applies to $G-e$, because $e$ is an original edge.
Conversely, restricting a proper coloring of $G^\triangle$ or
$G^\triangle-e$ to $V$ gives a proper coloring of $G$ or $G-e$,
respectively.
\end{proof}

Fix an original edge $e=ab\in E(G)$.  We now define the switch component
$C(G,e)$.  Every occurrence of a vertex in a different switch component is a
fresh copy.

\begin{definition}[Switch component]
\label{def:edge-switch-component}
Take a fresh copy of the core vertex set $V(G^\triangle)$.  For every edge
$f=xy\in E(G^\triangle)\setminus\{e\}$, introduce one ordinary subdivision
vertex $u_f$.  Introduce six further vertices
\[
  p_e,\quad s_e,\quad q_e,\quad z_e,\quad z_e^1,\quad z_e^2.
\]
The vertices $p_e$ and $q_e$ are the \emph{switch-path vertices}, $s_e$ is
the \emph{switch vertex}, and $z_e,z_e^1,z_e^2$ form the
\emph{switch-support triangle}.

Formally,
\[
\begin{split}
  V(C(G,e))={}&V(G^\triangle)
  \cup\{u_f:f\in E(G^\triangle)\setminus\{e\}\}\\
  &\cup\{p_e,s_e,q_e,z_e,z_e^1,z_e^2\}.
\end{split}
\]
The edge set consists of the following four classes.
\begin{enumerate}[label=(\roman*),leftmargin=2.4em]
  \item For every $f=xy\in E(G^\triangle)\setminus\{e\}$, add the two
  edges $\{x,u_f\}$ and $\{u_f,y\}$.  The edge $f$ itself is not retained.

  \item Replace the distinguished edge $e=ab$ by the switch path
  \[
    a-p_e-s_e-q_e-b.
  \]

  \item Add the distinguished switch edge
  \[
    \sigma_e=\{s_e,z_e\}.
  \]

  \item Add the switch-support triangle edges
  \[
    \{z_e,z_e^1\},\qquad \{z_e^1,z_e^2\},\qquad \{z_e^2,z_e\}.
  \]
\end{enumerate}
No other edges are present.
\end{definition}

The degrees in the component will be used repeatedly.  We record them
explicitly.

\begin{lemma}[Vertex degrees]
\label{lem:edge-component-degrees}
In $C(G,e)$, every ordinary subdivision vertex, every switch-path vertex,
and each of $z_e^1,z_e^2$ has degree two.  The vertices $s_e$ and $z_e$ have
degree three.  Every preprocessing vertex has degree two, and every original
vertex $v$ has degree $\deg_G(v)+2$.  Consequently,
\[
  \delta(C(G,e))=2
  \qquad\text{and}\qquad
  2\leq\dom(C(G,e))\leq3.
\]
\end{lemma}

\begin{proof}
An ordinary subdivision vertex $u_f$ is adjacent exactly to the two endpoints
of $f$.  The switch-path vertices satisfy
\[
  N(p_e)=\{a,s_e\},
  \qquad
  N(q_e)=\{s_e,b\}.
\]
Moreover,
\[
  N(s_e)=\{p_e,q_e,z_e\},
  \qquad
  N(z_e)=\{s_e,z_e^1,z_e^2\},
\]
and each of $z_e^1,z_e^2$ has its two neighbors in the support triangle.

Each preprocessing vertex is incident with exactly two private-triangle edges
in $G^\triangle$, both of which become subdivided paths.  An original vertex
$v$ is incident in $G^\triangle$ with its $\deg_G(v)$ original edges and its
two private-triangle edges.  Replacing one incident edge by the switch path
when $v\in\{a,b\}$ does not change this degree count.  Hence
$\deg_{C(G,e)}(v)=\deg_G(v)+2$.

Thus the minimum degree is two.  The upper bound follows from
Proposition~\ref{prop:degree-bound}; the lower bound follows from the absence
of isolated vertices and Proposition~\ref{prop:one-two-domatic}.
\end{proof}

\subsubsection{Local coloring analysis}
\label{subsec:edge-local-analysis}

Throughout this subsection the domatic colors are $\{1,2,3\}$.  For two
distinct colors $\alpha$ and $\beta$, write
$\operatorname{third}(\alpha,\beta)$ for the unique third color.

The first correspondence describes the switch while its distinguished edge is
present.  The switch can then absorb either relation between the colors of
$a$ and $b$, so the proper-coloring constraint represented by $e$ is absent.

\begin{lemma}[Switch present]
\label{lem:edge-switch-present}
For every graph $G$ and every original edge $e\in E(G)$,
\[
  \dom(C(G,e))=3
  \quad\Longleftrightarrow\quad
  G-e\text{ is $3$-colorable}.
\]
\end{lemma}

\begin{proof}
Suppose first that $G-e$ has a proper $3$-coloring $c$.  By
Lemma~\ref{lem:edge-saturation-colorability}, extend it to a proper
$3$-coloring of $G^\triangle-e$ in which every private triangle uses all
three colors.  For every ordinary edge
$f=xy\in E(G^\triangle)\setminus\{e\}$, set
\[
  c(u_f)=\operatorname{third}(c(x),c(y)).
\]

It remains to color the switch gadget.  If $c(a)\neq c(b)$, set
\[
  c(s_e)=\operatorname{third}(c(a),c(b)),
  \qquad
  c(p_e)=c(b),
  \qquad
  c(q_e)=c(a),
\]
and choose $c(z_e)$ arbitrarily.  If $c(a)=c(b)=\alpha$, choose the other two
colors $\beta$ and $\gamma$ and set
\[
  c(s_e)=\beta,
  \qquad
  c(p_e)=c(q_e)=\gamma,
  \qquad
  c(z_e)=\alpha.
\]
In either case, color $z_e^1$ and $z_e^2$ with the two colors different
from $c(z_e)$.

We verify every vertex type.
\begin{enumerate}[label=(\roman*),leftmargin=2.4em]
  \item For an ordinary subdivision vertex $u_f$, where $f=xy$, the closed
  neighborhood $\{x,u_f,y\}$ contains all three colors by construction.

  \item The closed neighborhoods
  \[
    N[p_e]=\{a,p_e,s_e\},
    \qquad
    N[q_e]=\{s_e,q_e,b\}
  \]
  contain all three colors in both switch cases.

  \item We have
  \[
    N[s_e]=\{s_e,p_e,q_e,z_e\}.
  \]
  If $c(a)\neq c(b)$, the vertices $s_e,p_e,q_e$ already use all three
  colors.  If $c(a)=c(b)$, the vertex $z_e$ supplies the unique color
  missing from $s_e,p_e,q_e$.

  \item The support vertices $z_e,z_e^1,z_e^2$ use the three distinct
  colors.  Thus $z_e^1$ and $z_e^2$ see all colors in their closed
  neighborhoods, and $z_e$ already sees all colors within the support
  triangle, independently of its additional neighbor $s_e$.

  \item Let $x$ be any core vertex.  In its private triangle, let $y$ and
  $w$ be the other two vertices.  Since $e$ is an original edge, neither
  private-triangle edge $xy$ nor $xw$ is distinguished.  The corresponding
  subdivision vertices satisfy
  \[
    c(u_{xy})=c(w),
    \qquad
    c(u_{xw})=c(y).
  \]
  Therefore $x$ sees its own color and the two other colors on these two
  subdivision vertices.  This verifies every original and every
  preprocessing core vertex.
\end{enumerate}
Hence $c$ is a $3$-domatic coloring of $C(G,e)$.  By
Lemma~\ref{lem:edge-component-degrees}, its domatic number is exactly three.

Conversely, suppose that $C(G,e)$ has a $3$-domatic coloring $c$.  For every
$f=xy\in E(G^\triangle)\setminus\{e\}$, the ordinary subdivision vertex
$u_f$ has degree two.  Lemma~\ref{lem:degree-two-forcing} gives
$c(x)\neq c(y)$.  The restriction of $c$ to $V(G^\triangle)$ is therefore a
proper $3$-coloring of $G^\triangle-e$.  Restricting further to $V(G)$ gives
a proper $3$-coloring of $G-e$.
\end{proof}

Deleting the switch edge separates the support triangle from the switch
vertex.  The switch vertex then has degree two, and degree-two forcing restores
the missing constraint corresponding to $e$.

\begin{lemma}[Switch deleted]
\label{lem:edge-switch-deleted}
For every graph $G$ and every original edge $e\in E(G)$,
\[
  \dom(C(G,e)-\sigma_e)=3
  \quad\Longleftrightarrow\quad
  G\text{ is $3$-colorable}.
\]
\end{lemma}

\begin{proof}
Suppose first that $G$ has a proper $3$-coloring $c$.  Extend it to a proper
$3$-coloring of $G^\triangle$, and color every ordinary subdivision vertex
with the third color of its endpoints.  Since $e=ab$ is properly colored,
$c(a)\neq c(b)$.  Set
\[
  c(s_e)=\operatorname{third}(c(a),c(b)),
  \qquad
  c(p_e)=c(b),
  \qquad
  c(q_e)=c(a).
\]
Color the support triangle with all three colors in any order.

Again we verify all vertex types.  Every ordinary subdivision vertex sees the
three colors on itself and its two endpoints.  The two switch-path vertices
see all three colors in $\{a,p_e,s_e\}$ and $\{s_e,q_e,b\}$, respectively.
After deletion of $\sigma_e$,
\[
  N[s_e]=\{p_e,s_e,q_e\},
\]
which contains all three colors.  Each support-triangle vertex sees all three
colors within that triangle.  Finally, the private-triangle argument from the
proof of Lemma~\ref{lem:edge-switch-present} shows that every original and
every preprocessing core vertex sees all three colors.  Thus
$C(G,e)-\sigma_e$ is $3$-domatic.

Conversely, suppose that $C(G,e)-\sigma_e$ has a $3$-domatic coloring $c$.
Every ordinary subdivision vertex forces the endpoints of its corresponding
edge in $G^\triangle-e$ to have different colors.  The switch vertex now has
degree two and
\[
  N[s_e]=\{p_e,s_e,q_e\},
\]
so $p_e,s_e,q_e$ have the three distinct colors.  Since
$N[p_e]=\{a,p_e,s_e\}$ also contains all three colors, $a$ has the color of
$q_e$.  Similarly, $b$ has the color of $p_e$.  Therefore
$c(a)\neq c(b)$.  The restriction to $V(G^\triangle)$ is a proper
$3$-coloring of all of $G^\triangle$, including the edge $e$, and its
restriction to $V(G)$ properly $3$-colors $G$.
\end{proof}

It remains to analyze the edges other than the distinguished switch edge.
Each such edge is deliberately incident with a vertex of degree two.  Deleting
it creates a leaf but never an isolated vertex.

\begin{lemma}[Non-switch edges]
\label{lem:edge-nonswitch-deletion}
For every $f\in E(C(G,e))\setminus\{\sigma_e\}$,
\[
  \dom(C(G,e)-f)=2.
\]
This holds independently of the colorability of $G-e$.
\end{lemma}

\begin{proof}
There are exactly three non-switch edge types.
\begin{enumerate}[label=(\roman*),leftmargin=2.4em]
  \item Each edge of an ordinary subdivided path $x-u_h-y$ is incident with
  the degree-two vertex $u_h$.  This includes the paths replacing edges of $G$
  other than $e$ and all paths replacing private-triangle edges.

  \item Each edge of the switch path $a-p_e-s_e-q_e-b$ is incident with one
  of the degree-two vertices $p_e,q_e$.

  \item Each edge of the switch-support triangle is incident with at least
  one of the degree-two vertices $z_e^1,z_e^2$.
\end{enumerate}
Deleting $f$ lowers such a degree-two endpoint to degree one.  The
minimum-degree bound therefore gives $\dom(C(G,e)-f)\leq2$.

No vertex becomes isolated.  The affected degree-two endpoint retains its
other incident edge.  Every core vertex had degree at least two before the
deletion, owing to its private triangle, and therefore retains degree at least
one.  The switch vertex and $z_e$ had degree three.  Hence the deleted graph
has no isolated vertex, and Proposition~\ref{prop:one-two-domatic} gives the
reverse inequality.  Thus its domatic number is exactly two.
\end{proof}

\begin{lemma}[The distinguished switch edge]
\label{lem:edge-switch-edge-deletion}
The graph $C(G,e)-\sigma_e$ has no isolated vertex and minimum degree two.
Consequently,
\[
  \dom(C(G,e)-\sigma_e)
  =
  \begin{cases}
    3,&\text{if $G$ is $3$-colorable},\\
    2,&\text{if $G$ is not $3$-colorable}.
  \end{cases}
\]
\end{lemma}

\begin{proof}
Deleting $\sigma_e$ lowers the degrees of $s_e$ and $z_e$ from three to two
and changes no other degree.  Thus the resulting graph has no isolated vertex
and has domatic number either two or three.  Lemma~\ref{lem:edge-switch-deleted}
determines which value occurs.
\end{proof}

\subsubsection{Correctness and DP-completeness}
\label{subsec:edge-reduction-correctness}

We now formally define the polynomial-time reduction $R$ outlined above by
assembling the switch components.

\begin{definition}[The reduction]
\label{def:edge-reduction-graph}
Let $G=(V,E)$ be a graph with $E\neq\varnothing$.  Define
\[
  R(G)=\mathop{\dot\bigcup}_{e\in E}C(G,e).
\]
All switch components use pairwise disjoint copies of every vertex introduced
in Definition~\ref{def:edge-switch-component}.
\end{definition}

For $n=|V(G)|$ and $m=|E(G)|$, each component $C(G,e)$ has
$O(n+m)$ vertices and edges and is constructible in polynomial time; since
$R(G)$ is the disjoint union of $m$ such components, the reduction $R$ is
computable in polynomial time.

\begin{proposition}[Correctness of the reduction]
\label{prop:edge-reduction-correctness}
For every graph $G$ with $E(G)\neq\varnothing$,
\[
\begin{aligned}
  G\in\EMinUncol
    &\Longleftrightarrow R(G)\in\DomCrit_3,\\
  R(G)\in\DomCrit_3
    &\Longleftrightarrow R(G)\in\DomCrit.
\end{aligned}
\]
Moreover, whenever these equivalent conditions hold,
$\dom(R(G))=3$.
\end{proposition}

\begin{proof}
Suppose first that $G\in\EMinUncol$.  Then $G$ is not $3$-colorable and
$G-e$ is $3$-colorable for every $e\in E(G)$.  In particular,
$E(G)\neq\varnothing$.  By Lemma~\ref{lem:edge-switch-present},
\[
  \dom(C(G,e))=3
  \qquad\text{for every }e\in E(G),
\]
and hence the disjoint-union formula gives
\[
  \dom(R(G))=3.
\]

Let $f$ be any edge of $R(G)$, and let $C(G,e)$ be the component containing
$f$.  If $f\neq\sigma_e$, then
Lemma~\ref{lem:edge-nonswitch-deletion} gives
\[
  \dom(C(G,e)-f)=2.
\]
If $f=\sigma_e$, then $G$ is not $3$-colorable, and
Lemma~\ref{lem:edge-switch-edge-deletion} again gives
\[
  \dom(C(G,e)-f)=2.
\]
Every other component remains $3$-domatic.  Therefore
\[
  \dom(R(G)-f)=2.
\]
Thus $R(G)\in\DomCrit_3$, and in particular $R(G)\in\DomCrit$.

Conversely, suppose that $G\notin\EMinUncol$.  We show that $R(G)$ is not
domatically critical.  Since $E(G)\neq\varnothing$, there are two cases.

First suppose that $G$ is $3$-colorable.
Then every graph $G-e$ is $3$-colorable.  Consequently, every component
$C(G,e)$ is $3$-domatic, and
\[
  \dom(R(G))=3.
\]
For every $e\in E(G)$,
Lemma~\ref{lem:edge-switch-edge-deletion} gives
\[
  \dom(C(G,e)-\sigma_e)=3,
\]
because $G$ is $3$-colorable.  Hence
\[
  \dom(R(G)-\sigma_e)=3=\dom(R(G)),
\]
so $R(G)$ is not domatically critical.

It remains that $G$ is not $3$-colorable but that $G-e$ is not
$3$-colorable for some $e\in E(G)$.  By
Lemmas~\ref{lem:edge-component-degrees} and
\ref{lem:edge-switch-present},
\[
  \dom(C(G,e))=2.
\]
Every switch component has domatic number at least two, and therefore
\[
  \dom(R(G))=2.
\]
Choose the non-switch support edge
\[
  f=\{z_e,z_e^1\}
\]
in the component $C(G,e)$.  By
Lemma~\ref{lem:edge-nonswitch-deletion},
\[
  \dom(C(G,e)-f)=2.
\]
All other components retain domatic number at least two, so
\[
  \dom(R(G)-f)=2=\dom(R(G)).
\]
Thus $R(G)$ is again not domatically critical.

These two cases exhaust all graphs with at least one edge outside
$\EMinUncol$, proving both equivalences.
\end{proof}

\begin{remark}[Disconnected and sparse input graphs]
\label{rem:edge-degenerate-sources}
The construction and the preceding proofs allow $G$ to be disconnected and to
contain isolated vertices.  Every vertex of $G$ receives a private triangle, so
its copy in every output component has degree at least two and
sees all three colors by the private-triangle mechanism.
Graphs with one edge or very few edges are
$3$-colorable and fall under the colorable no-instance case in the proof of
Proposition~\ref{prop:edge-reduction-correctness}.  Apart from the standing
assumption that the input has at least one edge, no structural promise beyond
finite simplicity is used.
\end{remark}

\Needspace{9\baselineskip}
\begin{theorem}
\label{thm:edge-target-three}
The problem $\DomCrit_3$ is DP-complete under polynomial-time many-one
reductions.
\end{theorem}

\begin{proof}
Membership follows from Proposition~\ref{prop:edge-fixed-membership}.  For
hardness, we reduce from the restriction of $\EMinUncol$ to graphs with at
least one edge, which remains DP-complete by the observation at the beginning
of this subsection.  By Proposition~\ref{prop:edge-reduction-correctness}, the
polynomial-time reduction $R$ satisfies
\[
  G\in\EMinUncol
  \quad\Longleftrightarrow\quad
  R(G)\in\DomCrit_3.
\]
Hence $\DomCrit_3$ is DP-hard and therefore DP-complete.
\end{proof}

\subsection{Extension to all fixed \texorpdfstring{{\boldmath $k\geq3$}}{k >= 3}}
\label{subsec:edge-lifting}

The classical clique-addition identity
\[
  \dom(G+K_r)=\dom(G)+r
\]
was established in Proposition~\ref{prop:clique-join}.  Zverovich and
Zverovich proved the corresponding structural theorem that joining a clique
preserves and reflects domatic criticality~\cite[Proposition~4]{ZverovichZverovich1991}.
We invoke this known result directly.

\begin{theorem}[Zverovich--Zverovich]
\label{thm:edge-clique-criticality}
Let $G$ be a graph and let $r\geq0$.  Then
\[
  G\text{ is domatically critical}
  \quad\Longleftrightarrow\quad
  G+K_r\text{ is domatically critical}.
\]
If, moreover, $G$ is domatically critical and $d=\dom(G)$, then every edge
$f$ of $G+K_r$ satisfies
\[
  \dom((G+K_r)-f)=d+r-1.
\]
\end{theorem}

\begin{proof}
For $r\geq1$, the equivalence is the cited result of Zverovich and
Zverovich.  The case $r=0$ is immediate.  Now suppose that $G$ is domatically
critical and put
$d=\dom(G)$.  By the equivalence, $G+K_r$ is domatically critical.  Hence,
for every edge $f$ of $G+K_r$,
\[
  \dom((G+K_r)-f)=\dom(G+K_r)-1.
\]
Proposition~\ref{prop:clique-join} gives $\dom(G+K_r)=d+r$, and therefore
\[
  \dom((G+K_r)-f)=d+r-1.
\]
\end{proof}

\begin{theorem}
\label{thm:edge-all-fixed-targets}
For every fixed integer $k\geq3$, the problem $\DomCrit_k$ is DP-complete.
\end{theorem}

\begin{proof}
Membership follows from Proposition~\ref{prop:edge-fixed-membership}.  For
hardness, put $r=k-3$ and map a graph $G$ to
\[
  R_k(G)=R(G)+K_r.
\]
Since $r$ is fixed, this reduction is polynomial-time computable.  By
Proposition~\ref{prop:clique-join} and
Theorem~\ref{thm:edge-clique-criticality},
\[
\begin{aligned}
  R_k(G)\in\DomCrit_k
  &\Longleftrightarrow R(G)\text{ is domatically critical}\\
  &\hspace{2.9em}\text{and }\dom(R(G))+r=k\\
  &\Longleftrightarrow R(G)\in\DomCrit_3\\
  &\Longleftrightarrow G\in\EMinUncol.
\end{aligned}
\]
The last equivalence is the reduction proved above.  Hence $\DomCrit_k$ is
DP-hard and therefore DP-complete.
\end{proof}

\subsection{The unrestricted problem}
\label{subsec:edge-unrestricted}

Recall from Definition~\ref{def:edge-criticality} that $\DomCrit$ contains
the graphs $H$ for which every edge deletion lowers $\dom(H)$ by one.

\begin{theorem}
\label{thm:edge-unrestricted-membership}
The unrestricted problem $\DomCrit$ is in $\ThetaTwoP$.
\end{theorem}

\begin{proof}
This is the standard bounded-query argument.  Let $H$ be an input graph.  If
$H$ is empty, then it can be handled directly, so suppose that $H$ is
nonempty and put
\[
  u=\delta(H)+1.
\]
By Proposition~\ref{prop:degree-bound},
\[
  1\leq \dom(H)\leq u.
\]
Moreover, the predicate $\dom(H)\geq t$ is monotone in $t$ and is an instance
of the domatic-number threshold problem.  We can therefore determine
$d=\dom(H)$ by binary search over the possible values
$t\in\{1,\ldots,u\}$.  This requires $O(\log(u+1))$ queries to an $\NP$
oracle.  Since $u\leq |V(H)|$, the number of queries is
$O(\log(|V(H)|+1))$.

Once $d$ is known, make one further $\NP$ query asking whether there exists
an edge $e\in E(H)$ such that $\dom(H-e)\geq d$.  A yes-certificate consists
of the edge $e$ together with a $d$-domatic coloring of $H-e$.

Accept precisely when this final query is answered no.  By the equivalent
strict-decrease formulation from Definition~2.2, this is exactly the condition
that every edge deletion lowers the domatic number.  The argument also covers
nonempty edgeless graphs, for which the final existential query is false.
Hence \textsc{DomCrit} is decidable with
\(O(\log(|V(H)|+1))\) adaptive queries to an \(\mathrm{NP}\) oracle
and therefore is in \(\Theta_2^p\).
\end{proof}

\begin{theorem}
\label{thm:edge-unrestricted-hardness}
The unrestricted problem $\DomCrit$ is DP-hard under polynomial-time many-one
reductions.
\end{theorem}

\begin{proof}
Apply $R$ to the restriction of $\EMinUncol$ to graphs with at least one
edge, which remains DP-complete by the observation at the beginning of
Subsection~\ref{subsec:edge-target-three}.  By the second equivalence in
Proposition~\ref{prop:edge-reduction-correctness}, the polynomial-time
reduction $R$ satisfies
\[
  G\in\EMinUncol
  \quad\Longleftrightarrow\quad
  R(G)\in\DomCrit.
\]
Since $\EMinUncol$ is DP-complete, $\DomCrit$ is DP-hard.
\end{proof}

\begin{corollary}
\label{cor:edge-unrestricted-bounds}
The unrestricted problem $\DomCrit$ is DP-hard and is in $\ThetaTwoP$.
\end{corollary}

\begin{remark}[Open problem]
\label{rem:edge-unrestricted-open}
It remains open whether $\DomCrit$ is $\ThetaTwoP$-hard and, consequently,
whether it is $\ThetaTwoP$-complete.
\end{remark}

%% file: conclusion.tex
\section{Conclusion and open problems}
\label{sec:conclusion}

We have determined the complexity of recognizing domatically critical graphs
when the domatic number is fixed.  At the two lowest target values, the
problem is polynomial-time decidable: $\DomCrit_1$ consists of the nonempty
edgeless graphs, while $\DomCrit_2$ consists precisely of the nonempty
disjoint unions of nontrivial stars.  From target value three onward, the
complexity changes sharply.  For every fixed integer $k\geq3$, the problem
$\DomCrit_k$ is $\DP$-complete under polynomial-time many-one reductions.
Thus, the fixed-target recognition problem admits the complete classification
\[
  \begin{aligned}
    \DomCrit_k &\in \Pclass
      &&\text{for } k\in\{1,2\},\\
    \DomCrit_k &\text{ is }\DP\text{-complete}
      &&\text{for every fixed } k\geq3.
  \end{aligned}
\]

The main technical step is the hardness proof at target value three.  The
switch construction translates edge-minimal $3$-uncolorability into the
requirement that every edge deletion lower the domatic number.  In particular,
the construction simultaneously controls the domatic number of the resulting
graph and the effect of deleting each of its different types of edges.  Once
the target-three case has been established, the clique-join lifting argument
raises the domatic number while preserving and reflecting domatic criticality.
Together, these two mechanisms yield the classification for all fixed target
values.

For the unrestricted recognition problem, we have shown that
\[
  \DomCrit\text{ is }\DP\text{-hard}
  \qquad\text{and}\qquad
  \DomCrit\in\ThetaTwoP.
\]
The exact complexity of this problem remains open.  In particular, it is not
known whether $\DomCrit$ is $\ThetaTwoP$-complete; the missing direction is
$\ThetaTwoP$-hardness.

%% file: references.bib
@article{BurjonsEtAl2024,
  author  = {Burjons, Elisabet and Frei, Fabian and Hemaspaandra, Edith and Komm, Dennis and Wehner, David},
  title   = {Finding Optimal Solutions with Neighborly Help},
  journal = {Algorithmica},
  volume  = {86},
  number  = {6},
  year    = {2024},
  pages   = {1921--1947},
  doi     = {10.1007/s00453-023-01204-1},
  url     = {https://doi.org/10.1007/s00453-023-01204-1}
}

@article{Rall1990,
  author  = {Rall, Douglas F.},
  title   = {Domatically Critical and Domatically Full Graphs},
  journal = {Discrete Mathematics},
  volume  = {86},
  number  = {1--3},
  year    = {1990},
  pages   = {81--87},
  doi     = {10.1016/0012-365X(90)90351-H},
  url     = {https://doi.org/10.1016/0012-365X(90)90351-H}
}

@article{ZverovichZverovich1991,
  author  = {Zverovich, Igor Edmundovich and Zverovich, Vadim E.},
  title   = {A Note on Domatically Critical and Cocritical Graphs},
  journal = {Czechoslovak Mathematical Journal},
  volume  = {41},
  number  = {2},
  year    = {1991},
  pages   = {278--281},
  doi     = {10.21136/CMJ.1991.102460},
  url     = {https://dml.cz/handle/10338.dmlcz/102460}
}

@article{Zelinka1980,
  author  = {Zelinka, Bohdan},
  title   = {Domatically Critical Graphs},
  journal = {Czechoslovak Mathematical Journal},
  volume  = {30},
  number  = {3},
  year    = {1980},
  pages   = {486--489},
  doi     = {10.21136/CMJ.1980.101697},
  url     = {https://dml.cz/handle/10338.dmlcz/101697}
}

@article{CockayneHedetniemi1977,
  author  = {Cockayne, Ernest J. and Hedetniemi, Stephen T.},
  title   = {Towards a Theory of Domination in Graphs},
  journal = {Networks},
  volume  = {7},
  number  = {3},
  year    = {1977},
  pages   = {247--261},
  doi     = {10.1002/net.3230070305},
  url     = {https://doi.org/10.1002/net.3230070305}
}

@incollection{Cockayne1978Survey,
  author    = {Cockayne, Ernest J.},
  title     = {Domination of Undirected Graphs---A Survey},
  editor    = {Alavi, Yousef and Lick, Don R.},
  booktitle = {Theory and Applications of Graphs},
  series    = {Lecture Notes in Mathematics},
  volume    = {642},
  publisher = {Springer},
  address   = {Berlin, Heidelberg},
  year      = {1978},
  pages     = {141--147},
  doi       = {10.1007/BFb0070371},
  url       = {https://doi.org/10.1007/BFb0070371}
}

@article{CaiMeyer1987,
  author  = {Cai, Jin-Yi and Meyer, Gabriele E.},
  title   = {Graph Minimal Uncolorability Is {$D^P$}-Complete},
  journal = {SIAM Journal on Computing},
  volume  = {16},
  number  = {2},
  year    = {1987},
  pages   = {259--277},
  doi     = {10.1137/0216022},
  url     = {https://doi.org/10.1137/0216022}
}

@article{PapadimitriouYannakakis1984,
  author  = {Papadimitriou, Christos H. and Yannakakis, Mihalis},
  title   = {The Complexity of Facets (and Some Facets of Complexity)},
  journal = {Journal of Computer and System Sciences},
  volume  = {28},
  number  = {2},
  year    = {1984},
  pages   = {244--259},
  doi     = {10.1016/0022-0000(84)90068-0},
  url     = {https://doi.org/10.1016/0022-0000(84)90068-0}
}

@article{PapadimitriouWolfe1988,
  author  = {Papadimitriou, Christos H. and Wolfe, David},
  title   = {The Complexity of Facets Resolved},
  journal = {Journal of Computer and System Sciences},
  volume  = {37},
  number  = {1},
  year    = {1988},
  pages   = {2--13},
  doi     = {10.1016/0022-0000(88)90042-6},
  url     = {https://doi.org/10.1016/0022-0000(88)90042-6}
}

@article{CaiEtAl1988BooleanI,
  author  = {Cai, Jin-Yi and Gundermann, Thomas and Hartmanis, Juris and Hemachandra, Lane A. and Sewelson, Vivian and Wagner, Klaus W. and Wechsung, Gerd},
  title   = {The Boolean Hierarchy {I}: Structural Properties},
  journal = {SIAM Journal on Computing},
  volume  = {17},
  number  = {6},
  year    = {1988},
  pages   = {1232--1252},
  doi     = {10.1137/0217078},
  url     = {https://doi.org/10.1137/0217078}
}

@article{CaiEtAl1989BooleanII,
  author  = {Cai, Jin-Yi and Gundermann, Thomas and Hartmanis, Juris and Hemachandra, Lane A. and Sewelson, Vivian and Wagner, Klaus W. and Wechsung, Gerd},
  title   = {The Boolean Hierarchy {II}: Applications},
  journal = {SIAM Journal on Computing},
  volume  = {18},
  number  = {1},
  year    = {1989},
  pages   = {95--111},
  doi     = {10.1137/0218007},
  url     = {https://doi.org/10.1137/0218007}
}

@article{KoblerSchoningWagner1987,
  author  = {K{\"o}bler, Johannes and Sch{\"o}ning, Uwe and Wagner, Klaus W.},
  title   = {The Difference and Truth-Table Hierarchies for {NP}},
  journal = {RAIRO---Theoretical Informatics and Applications},
  volume  = {21},
  number  = {4},
  year    = {1987},
  pages   = {419--435},
  doi     = {10.1051/ita/1987210404191},
  url     = {https://doi.org/10.1051/ita/1987210404191}
}

@article{Wagner1987,
  author  = {Wagner, Klaus W.},
  title   = {More Complicated Questions about Maxima and Minima, and Some Closures of {NP}},
  journal = {Theoretical Computer Science},
  volume  = {51},
  number  = {1--2},
  year    = {1987},
  pages   = {53--80},
  doi     = {10.1016/0304-3975(87)90049-1},
  url     = {https://doi.org/10.1016/0304-3975(87)90049-1}
}

@article{Wagner1990,
  author  = {Wagner, Klaus W.},
  title   = {Bounded Query Classes},
  journal = {SIAM Journal on Computing},
  volume  = {19},
  number  = {5},
  year    = {1990},
  pages   = {833--846},
  doi     = {10.1137/0219058},
  url     = {https://doi.org/10.1137/0219058}
}

@article{RiegeRothe2006,
  author  = {Riege, Tobias and Rothe, J{\"o}rg},
  title   = {Complexity of the Exact Domatic Number Problem and of the Exact Conveyor Flow Shop Problem},
  journal = {Theory of Computing Systems},
  volume  = {39},
  number  = {5},
  year    = {2006},
  pages   = {635--668},
  doi     = {10.1007/s00224-004-1209-8},
  url     = {https://doi.org/10.1007/s00224-004-1209-8}
}

@article{RiegeRothe2006Survey,
  author  = {Riege, Tobias and Rothe, J{\"o}rg},
  title   = {Completeness in the Boolean Hierarchy: Exact-Four-Colorability, Minimal Graph Uncolorability, and Exact Domatic Number Problems---A Survey},
  journal = {Journal of Universal Computer Science},
  volume  = {12},
  number  = {5},
  year    = {2006},
  pages   = {551--578},
  url     = {https://lib.jucs.org/article/28616/}
}

@article{FreiHemaspaandraRothe2022,
  author  = {Frei, Fabian and Hemaspaandra, Edith and Rothe, J{\"o}rg},
  title   = {Complexity of Stability},
  journal = {Journal of Computer and System Sciences},
  volume  = {123},
  year    = {2022},
  pages   = {103--121},
  doi     = {10.1016/j.jcss.2021.07.001},
  url     = {https://doi.org/10.1016/j.jcss.2021.07.001}
}

@book{GareyJohnson1979,
  author    = {Garey, Michael R. and Johnson, David S.},
  title     = {Computers and Intractability: A Guide to the Theory of NP-Completeness},
  publisher = {W. H. Freeman},
  address   = {San Francisco},
  year      = {1979}
}

@misc{Spakowski2026ExactDomatic,
  author        = {Spakowski, Holger},
  title         = {Closing the Complexity Gap for Exact Domatic Number at Three and Four},
howpublished  = {arXiv:2607.09442 [cs.CC]},
year          = {2026},
  eprint        = {2607.09442},
  archivePrefix = {arXiv},
  primaryClass  = {cs.CC},
  doi           = {10.48550/arXiv.2607.09442},
  url           = {https://arxiv.org/abs/2607.09442}
}

@article{KaplanShamir1994,
  author  = {Kaplan, Haim and Shamir, Ron},
  title   = {The Domatic Number Problem on Some Perfect Graph Families},
  journal = {Information Processing Letters},
  volume  = {49},
  number  = {1},
  year    = {1994},
  pages   = {51--56},
  doi     = {10.1016/0020-0190(94)90054-X},
  url     = {https://doi.org/10.1016/0020-0190(94)90054-X}
}

@inproceedings{Hemachandra1987,
  author    = {Lane A. Hemachandra},
  title     = {The Strong Exponential Hierarchy Collapses},
  booktitle = {Proceedings of the Nineteenth Annual ACM Symposium on
               Theory of Computing},
  series    = {STOC '87},
  pages     = {110--122},
  publisher = {ACM},
  year      = {1987},
  doi       = {10.1145/28395.28408},
  url       = {https://doi.org/10.1145/28395.28408}
}

@inproceedings{BussHay1988,
  author    = {Samuel R. Buss and Louise Hay},
  title     = {On Truth-Table Reducibility to {SAT} and the Difference
               Hierarchy over {NP}},
  booktitle = {Proceedings of the Third Annual Structure in Complexity
               Theory Conference},
  pages     = {224--233},
  publisher = {IEEE Computer Society},
  year      = {1988},
  doi       = {10.1109/SCT.1988.5282},
  url       = {https://doi.org/10.1109/SCT.1988.5282}
}
